\documentclass[10pt]{amsart}

\usepackage{amsmath}
\usepackage{graphicx} 

\newtheorem{theorem}{Theorem}[section]
\newtheorem{lemma}[theorem]{Lemma}
\newtheorem{proposition}[theorem]{Proposition}

\newcommand{\e}{\epsilon}
\newcommand{\z}{\zeta}
\newcommand{\C}{\mathbb{C}}
\newcommand{\R}{\mathbb{R}}
\newcommand{\V}{\mathcal{V}}
\newcommand{\abs}[1]{\lvert#1\rvert}
\newcommand{\bigabs}[1]{\big\lvert#1\big\rvert}
\newcommand{\Bigabs}[1]{\Big\lvert#1\Big\rvert}
\newcommand{\biggabs}[1]{\bigg\lvert#1\bigg\rvert}
\newcommand{\Biggabs}[1]{\Bigg\lvert#1\Bigg\rvert}
\newcommand{\leg}[2]{({#1}\!\mid\!{#2})}
\newcommand{\stack}[2]{\genfrac{}{}{0pt}{}{#1}{#2}}
\DeclareMathOperator{\GF}{GF}
\DeclareMathOperator{\Tr}{Tr}

\begin{document}

\title[The merit factor of binary arrays]
      {The merit factor of binary arrays\\derived from the quadratic character}

\author{Kai-Uwe Schmidt}

\date{30 June 2010 (revised 18 July 2011)}

\address{Department of Mathematics, 
Simon Fraser University, 8888 University Drive, Burnaby BC V5A 1S6, Canada.}

\email{kuschmidt@sfu.ca}

\thanks{The author is supported by Deutsche Forschungsgemeinschaft (German Research Foundation) under Research Fellowship SCHM 2609/1-1}

\subjclass[2010]{Primary: 94A55; Secondary: 68P30, 05B10}

\begin{abstract}
We calculate the asymptotic merit factor, under all cyclic rotations of rows and columns, of two families of binary two-dimensional arrays derived from the quadratic character. The arrays in these families have size $p\times q$, where $p$ and $q$ are not necessarily distinct odd primes, and can be considered as two-dimensional generalisations of a Legendre sequence. The asymptotic values of the merit factor of the two families are generally different, although the maximum asymptotic merit factor, taken over all cyclic rotations of rows and columns, equals $36/13$ for both families. These are the first non-trivial theoretical results for the asymptotic merit factor of families of truly two-dimensional binary arrays.
\end{abstract}

\maketitle

\section{Introduction}
\label{sec:intro}

We consider an \emph{array of size $n\times m$} to be an infinite matrix $A=(a_{ij})$ of real-valued elements satisfying
\[
a_{ij}=0\quad\mbox{unless $0\le i<n$ and $0\le j<m$}.
\]
The array is called \emph{binary} if $a_{ij}$ takes values only in $\{1,-1\}$ for all $i,j$ satisfying $0\le i<n$ and $0\le j<m$, and is called \emph{ternary} if $a_{ij}$ takes values only in $\{0,1,-1\}$. Given integers $u$ and $v$, the \emph{aperiodic autocorrelation} of $A=(a_{ij})$ at displacement $(u,v)$ is defined to be
\[
C_A(u,v):=\sum_{i,j}a_{ij}a_{i+u,j+v}.
\]
We refer to an array of size $n\times 1$ as a \emph{sequence of length $n$}, abbreviating the array $(a_{i0})$ to $(a_i)$ and the aperiodic autocorrelation $C_A(u,0)$ to $C_A(u)$.
\par
Binary arrays with small out-of-phase aperiodic autocorrelation have a wide range of applications in digital communications and storage systems, including radar~\cite{Alquaddoomi1989} and steganography~\cite{Schyndel1999}. Ideally, we would like to find a binary array $A$ of size $n\times m$ satisfying
\begin{equation}
\abs{C_A(u,v)}\le 1\quad\mbox{for all $(u,v)\ne(0,0)$},   \label{eqn:Barker}
\end{equation}
in which case, $A$ is called a \emph{Barker array}~\cite{Alquaddoomi1989}. However, it was recently shown by Davis, Jedwab, and Smith~\cite{Davis2007} that a (truly two-dimensional) Barker array must have size $2\times 2$. (Barker sequences of length~$n$, namely $n\times 1$ Barker arrays, are known for $n\in\{2,3,4,5,7,11,13\}$, and any other Barker sequence must have even length~\cite{Turyn1961} greater than $10^{29}$~\cite{Mossinghoff2009}.)
\par
Since the Barker array criterion~\eqref{eqn:Barker} is too restrictive for array dimensions exceeding $2\times 2$, it is natural to define a measure for the collective smallness of the aperiodic autocorrelation values of a binary array. One such measure is the \emph{merit factor}, which is defined for a binary array $A=(a_{ij})$ of size $n\times m$ with $nm>1$ to be
\[
F(A):=\frac{(nm)^2}{\sum_{(u,v)\ne (0,0)}[C_A(u,v)]^2}.
\]
Let $F_{n,m}$ denote the maximum value of $F(A)$ taken over all $2^{nm}$ binary arrays $A$ of size $n\times m$, and abbreviate $F_{n,1}$ to $F_n$. We note that the mean of $1/F(A)$, taken over all binary sequences $A$ of length $n$, equals $1-1/n$~\cite{Sarwate1984}. The argument of~\cite{Sarwate1984} easily generalises to two dimensions: the mean of $1/F(A)$, taken over all $2^{nm}$ binary arrays $A$ of size $n\times m$, equals $1-1/(nm)$. It follows that $F_{n,m}\ge nm/(nm-1)$, which asymptotically equals~$1$ and provides a first benchmark result.
\par
A number of theoretical and computational results on $F_n$ are known (see~\cite{Jedwab2005} for a survey). One line of research is to calculate $F_n$ for small values of~$n$. At present, $F_n$ has been calculated for all $n\le 60$ (see~\cite[Fig.~1]{Jedwab2005}, for example). The largest values of $F_n$ currently known are $F_{13}=\frac{169}{12}\simeq14.08$ and $F_{11}=\frac{121}{10}$, which are attained by Barker sequences. The computational analysis of $F_n$ quickly becomes infeasible as $n$ grows. Another line of research is therefore to construct particular infinite families of binary sequences of increasing length and to calculate their asymptotic merit factor.
\par
The only non-trivial infinite families of binary sequences for which the asymptotic value of the merit factor is known are Rudin-Shapiro sequences \cite{Littlewood1968}, Legendre sequences \cite{Hoholdt1988}, and $m$-sequences \cite{Jensen1989}, together with some generalisations of these families \cite{Hoholdt1985}, \cite{Jensen1991}, \cite{Schmidt2009}, \cite{Jedwab2010}, \cite{Jedwab2010a}. The largest proven asymptotic merit factor of a binary sequence family is~$6$, which is attained by cyclically rotated Legendre sequences (see Theorem~\ref{thm:hj}). There is also considerable numerical evidence, though currently no proof, that an asymptotic merit factor greater than~$6.34$ can be achieved for a family of binary sequences related to Legendre sequences~\cite{Borwein2004}.
\par
Much less is known about the value of $F_{n,m}$ for $n,m>1$. Eggers~\cite[Tab.~4.2]{Eggers1986} computed $F_{n,m}$ for $nm\le 21$ and found lower bounds on $F_{n,m}$ for $nm\le 121$. Although the data supplied in~\cite{Eggers1986} are very limited, it is apparent that $F_{n,m}$ tends to be smaller than $F_{nm}$. The largest value of $F_{n,m}$ for $n,m>1$ reported in~\cite{Eggers1986} equals $F_{4,4}=\frac{16}{3}\simeq5.33$. However, an elementary construction technique gives binary arrays with larger merit factor. Given two sequences $A=(a_i)$ and $B=(b_j)$ of length $n$ and $m$, respectively, we follow~\cite{Calabro1968} in defining the \emph{product array} $A\times B:=(a_ib_j)$. A straightforward calculation shows that
\begin{equation}
C_{A\times B}(u,v)=C_A(u)C_B(v),   \label{eqn:C_product}
\end{equation}
from which we deduce
\begin{equation}
\frac{1}{F(A\times B)}=\Big(1+\frac{1}{F(A)}\Big)\Big(1+\frac{1}{F(B)}\Big)-1.   \label{eqn:F_prod}
\end{equation}
Let $A$ and $B$ be Barker sequences of length~$13$ and $11$, respectively. It follows from~\eqref{eqn:F_prod} that $F(A\times A)\simeq6.80$, $F(A\times B)\simeq6.27$, and $F(B\times B)\simeq5.81$. Another consequence of~\eqref{eqn:F_prod} is
\[
F_{n,m}\ge \frac{F_nF_m}{F_n+F_m+1}.
\]
\par
Currently, no theoretical results on the asymptotic merit factor of families of binary truly two-dimensional arrays are known. B\"omer and Antweiler~\cite{Boemer1993a} analysed the merit factor of several binary array families numerically. Among the investigated families, two types of array families related to the quadratic character appeared to have largest merit factor. The arrays in the first family were proposed by Calabro and Wolf~\cite{Calabro1968} and have size $p\times q$, and the arrays in the second family were proposed by B\"omer, Antweiler, and Schotten~\cite{Boemer1993} and have size $p\times p$, where $p$ and $q$ are (not necessarily distinct) odd primes. Both families can be considered as two-dimensional generalisations of Legendre sequences. The authors of~\cite{Boemer1993a} successively applied three operations, namely rotations of rows and columns, stairlike rotations of rows and columns, and proper decimations, and computed the maximum value of the merit factor for arrays of small sizes taken from these families. They then remarked~\cite[p.~8]{Boemer1993a} that
\begin{quote}
``\dots, for large arrays, the ACF [aperiodic autocorrelation function] merit factors of both classes [the above mentioned array families of square size] appear to tend to~$3$.'',
\end{quote}
and asked for a theoretical explanation of this observation.
\par
In this paper we study the merit factor of two families of binary arrays. The arrays in the first family, called \emph{Legendre arrays}, have size $p\times q$, where $p$ and $q$ are (not necessarily distinct) odd primes, and contain as a special case the arrays proposed by Calabro and Wolf~\cite{Calabro1968}. The arrays in the second family, called \emph{quadratic-residue arrays}, have size $p\times p$, where $p$ is an odd prime, and contain as a special case the arrays proposed by B\"omer, Antweiler, and Schotten~\cite{Boemer1993}. We calculate, under certain conditions on the growth rate of $p$ relative to $q$, the asymptotic merit factor at all rotations of rows and columns for both array families. In particular, we show that for both families the asymptotic merit factor equals $\frac{36}{13}\simeq 2.77$ for an optimal rotation of rows and columns. Although we only maximise the merit factor with respect to the first operation considered in~\cite{Boemer1993a}, namely rotations of rows and columns, this result does not support the conclusion of~\cite{Boemer1993} quoted above. For all other (non-optimal) rotations of rows and columns, the asymptotic merit factor of quadratic-residue arrays is larger than that of square Legendre arrays.


\section{Two Families of Binary Arrays}

Given an odd prime $p$ and a positive integer $m$, let $\GF(p^m)$ be the finite field containing $p^m$ elements. Whenever convenient, we treat integers after reduction modulo $p$ as elements in $\GF(p)$. The \emph{quadratic character of $\GF(p^m)$} is the function $\chi:\GF(p^m)\to\R$ defined by
\[
\chi(a):=\begin{cases}
0 & \mbox{for $a=0$}\\
-1 & \mbox{for $a$ not a square in $\GF(p^m)$}\\
+1 & \mbox{otherwise}.
\end{cases}
\]
This function is multiplicative:
\begin{equation}
\label{eqn:multiplicativity}
\chi(a)\chi(b)=\chi(ab).
\end{equation}
If $m=1$, then $\leg{a}{p}:=\chi(a)$ is the \emph{Legendre symbol}. A \emph{Legendre sequence} $L=(\ell_i)$ of prime length $p>2$ is defined by
\[
\ell_i:=\begin{cases}
1          & \mbox{for $i=0$}\\
\leg{i}{p} & \mbox{for $1\le i<p$}.
\end{cases}
\]
If the initial element in a Legendre sequence is changed to zero, so that
\[
\ell_i=\leg{i}{p}\quad\mbox{for $0\le i<p$},
\]
then we call $L$ a \emph{ternary Legendre sequence}. 
\par
In what follows, we present two families of binary arrays, which can be considered as two-dimensional generalisations of Legendre sequences. Let $p$ and $q$ be two (not necessarily distinct) odd primes, and let $\V_{p,q}$ be the set of ternary arrays $V=(v_{ij})$ of size $p\times q$ satisfying
\[
\abs{v_{ij}}=\begin{cases}
1 & \quad\mbox{for ($i=0$ and $0\le j<q$) or ($j=0$ and $0\le i<p$)}\\
0 & \quad\mbox{otherwise}.
\end{cases}
\]
Then $\V_{p,q}$ contains $2^{p+q-1}$ arrays, each having $p+q-1$ nonzero elements. Given ternary Legendre sequences $L$ and $K$ of length $p$ and $q$, respectively, we define a~\emph{Legendre array} $X$ of size $p\times q$ to be a binary array of size $p\times q$ that can be written as
\[
X=L\times K+V\quad\mbox{for some $V\in\V_{p,q}$}.
\]
For example, the array $X=(x_{ij})$ of size $p\times q$, given by
\begin{equation*}
x_{ij}:=\begin{cases}
-1                   & \mbox{for $j=0$ and $0\le i<p$}\\
+1                   & \mbox{for $i=0$ and $1\le j<q$}\\
\leg{i}{p}\leg{j}{q} & \mbox{for $1\le i<p$ and $1\le j<q$},
\end{cases}
\end{equation*}
is a Legendre array of size $p\times q$. This particular array was originally defined by Calabro and Wolf~\cite{Calabro1968}, and its merit factor properties were investigated numerically in~\cite{Boemer1993a}. In the original paper~\cite{Calabro1968} such an array was called a ``quadratic-residue array''. We use the term \emph{Legendre array} to  distinguish it from our second family of binary arrays.
\par
Let $p$ be an odd prime, let $\chi$ be the quadratic character of $\GF(p^2)$, and let $\{\alpha,\alpha'\}$ be a basis for $\GF(p^2)$ over $\GF(p)$. Following B\"omer, Antweiler, and Schotten~\cite{Boemer1993}, we define a \emph{quadratic-residue array} $Y=(y_{ij})$ of size $p\times p$ to be a binary array of size $p\times p$ satisfying
\[
y_{ij}:=\begin{cases}
\mbox{$+1$ or $-1$}    & \mbox{for $i=j=0$}\\
\chi(i\alpha+j\alpha') & \mbox{for $0\le i,j<p,\,(i,j)\ne (0,0)$}.
\end{cases}
\]
The class of quadratic-residue arrays $(y_{ij})$ satisfying $y_{00}=+1$ was defined by B\"omer, Antweiler, and Schotten~\cite{Boemer1993}, and its merit factor properties were investigated numerically in~\cite{Boemer1993a}. In our analysis it will be convenient to change the leading element in a quadratic-residue array to zero. Accordingly, we define the \emph{ternary quadratic-residue array} of size $p\times p$ to be the ternary array $Z=(z_{ij})$ of size $p\times p$ given by
\begin{equation}
\label{eqn:QRarray2}
z_{ij}:=\chi(i\alpha+j\alpha')\quad\mbox{for $0\le i,j<p$}.
\end{equation}
\par
Next we define an operation acting on an array to produce a new array of the same size. Given an array $A=(a_{ij})$ of size $n\times m$ and real numbers $s$ and $t$, the \emph{rotation} $A_{s,t}$ is the array $B=(b_{ij})$ of size $n\times m$ given by
\begin{equation}
b_{ij}=a_{(i+\lfloor ns\rfloor)\bmod n,\,(j+\lfloor mt\rfloor)\bmod m}\quad\mbox{for $0\le i<n$ and $0\le j<m$}.   \label{eqn:rotation}
\end{equation}
If $A$ is a sequence of length~$n$, we abbreviate $A_{s,0}$ to $A_s$.
\par
The asymptotic merit factor of a Legendre sequence was calculated for all rotations by H{\o}holdt and Jensen~\cite{Hoholdt1988}.
\begin{theorem}[H{\o}holdt and Jensen~\cite{Hoholdt1988}]
\label{thm:hj}
Let $L$ be the Legendre sequence of prime length $p>2$, and let $s$ be a real number satisfying $-\tfrac{1}{2}<s\le\frac{1}{2}$. Then
\[
\frac{1}{\lim\limits_{p\to\infty}F(L_s)}=
\tfrac{1}{6}+8\left(\abs{s}-\tfrac{1}{4}\right)^2.
\]
\end{theorem}
\par
The constraint $-\tfrac{1}{2}<s\le\tfrac{1}{2}$ in Theorem~\ref{thm:hj} is for notational convenience only since by definition $A_s$ is the same as $A_{s+1}$ for every sequence $A$ and all real $s$. The maximum asymptotic merit factor of a rotated Legendre sequence~$L_s$ is~$6$, which occurs for $s=\pm\tfrac{1}{4}$.


\section{Calculation of the Merit Factor of an Array}

Given a positive integer $n$, let
\[
\z_n:=e^{\sqrt{-1}\,\pi /n}
\]
be a primitive $(2n)$th root of unity. Let $A=(a_{ij})$ be an array of size $n\times m$. The \emph{generating function} of $A$ is defined to be the power series
\[
A(x,y):=\sum_{i,j}a_{ij}x^iy^j.
\]
If $A$ is a sequence of length $n$, we write $A(x)$ for $A(x,y)$.
\par
The next lemma shows how the merit factor of $A$ can be computed from the values $A(\z_n^i,\z_m^j)$. This approach generalises to two dimensions the method of H{\o}holdt and Jensen~\cite{Hoholdt1988} to compute the asymptotic merit factor of a sequence.
\begin{lemma}
\label{lem:CQ}
Let $A$ be an array of size $n\times m$. Then
\[
\sum_{u,v}[C_A(u,v)]^2=\frac{1}{4nm}\sum_{i=0}^{2n-1}\sum_{j=0}^{2m-1}\bigabs{A(\z_n^i,\z_m^j)}^4.
\]
\end{lemma}
\begin{proof}
Straightforward manipulation shows that
\begin{align}
A(x,y)A(x^{-1},y^{-1})&=\sum_{u,v}C_A(u,v)x^{-u}y^{-v}\quad\mbox{for $x\ne 0$ and $y\ne 0$},   \nonumber
\intertext{and therefore,}
\bigabs{A(x,y)}^2&=\sum_{u,v}C_A(u,v)x^{-u}y^{-v}\quad\mbox{for $\abs{x}=\abs{y}=1$}.   \label{eqn:abs_A_poly}
\end{align}
Using this identity, an elementary calculation gives
\begin{align*}
\frac{1}{4nm}\sum_{i=0}^{2n-1}\sum_{j=0}^{2m-1}\bigabs{A(\zeta_n^i,\zeta_m^j)}^4&=\frac{1}{4nm}\sum_{i=0}^{2n-1}\sum_{j=0}^{2m-1}\bigabs{A(\zeta_n^i,\zeta_m^j)}^2\overline{\bigabs{A(\zeta_n^i,\zeta_m^j)}^2}\\
&=\sum_{u,v}[C_A(u,v)]^2,
\end{align*}
as required.
\end{proof}


\section{The Merit Factor of Legendre Arrays}

In this section we compute the asymptotic merit factor of Legendre arrays for all rotations, subject to certain conditions on the growth rate of the dimensions. We first record a result on the aperiodic autocorrelation of rotated ternary Legendre sequences, which arises as an immediate corollary of~\cite[Thm.~3]{Schmidt2009}.
\begin{proposition}
\label{prop:FLegendre}
Let $L$ be the ternary Legendre sequence of prime length $p>2$, and let $s$ be a real number satisfying $-\tfrac{1}{2}<s\le\tfrac{1}{2}$. Then, as $p\to\infty$,
\[
\frac{1}{p^2}\sum_{u}\;[C_{L_s}(u)]^2=\tfrac{7}{6}+8\left(\abs{s}-\tfrac{1}{4}\right)^2+O(p^{-1}(\log p)^2).
\]
\end{proposition}
\par
We note that, as explained after~\cite[Thm.~3]{Schmidt2009}, we can recover Theorem~\ref{thm:hj} from Proposition~\ref{prop:FLegendre}. We also need the following bound for the magnitude of a polynomial over $\C$ at a $(2d)$th root of unity, in terms of its magnitudes at $d$th roots of unity.
\begin{lemma}
\label{lem:poly-bound}
Let $d>1$ be odd, and let $A\in\C[x]$ have degree at most $d-1$.
Then
\[
\bigabs{A(\z_d^j)}\le(2\log d)\max_{0\le k<d}\,\bigabs{A(\z_d^{2k})}\quad\mbox{for integer $j$}.
\]
\end{lemma}
\begin{proof}
The bound is trivial in the case that $j$ is even. We may therefore take $j$ to be odd, writing $j=2\ell+d$ for some integer~$\ell$ so that $\z_d^j=-\z_d^{2\ell}$. It is then sufficient to bound $\abs{A(-\z_d^{2\ell})}$. Now by Lagrange interpolation we have
\[
A(x)=\frac{1}{d}\sum_{k=0}^{d-1}\frac{x^d-1}{x-\z_d^{2k}}\,\z_d^{2k}\,A(\z_d^{2k}),
\]
and so, since $d$ is odd,
\begin{align*}
\bigabs{A(-\z_d^{2\ell})}&\le\frac{1}{d}\sum_{k=0}^{d-1}\frac{2}{\abs{\z_d^{2\ell}+\z_d^{2k}}}\bigabs{A(\z_d^{2k})}\\
&\le\frac{2}{d}\,\max_{0\le k<d}\,\bigabs{A(\z_d^{2k})}\sum_{m=0}^{d-1}\frac{1}{\abs{1+\z_d^{2m}}}.
\end{align*}
The result follows from the inequality
$\sum_{m=0}^{d-1}1/\abs{1+\z_d^{2m}}\le d\log d$ (which holds since $d$ is odd, see~\cite[p.~625]{Jensen1991}, for example).
\end{proof}
\par
The next theorem gives, under certain conditions on the growth rate of $p$ relative to $q$, the asymptotic merit factor of all $2^{p+q-1}$ Legendre arrays of size $p\times q$ for all rotations.
\begin{theorem}
\label{thm:Farray1}
Let $\mathcal{N}$ be an infinite set of products of two not necessarily distinct odd primes, and let $N$ take values only in $\mathcal{N}$. Write $N=pq$ for odd primes~$p$ and~$q$, and suppose that
\begin{equation}
\label{eqn:cond_pq}
\frac{q}{p^2}\to0\quad\mbox{and}\quad
\frac{p}{q^2}\to0\quad\mbox{as $N\to\infty$}.
\end{equation}
Let $X$ be a Legendre array of size $p\times q$, and let $s$ and $t$ be real numbers satisfying $-\tfrac{1}{2}<s,t\le\tfrac{1}{2}$. Then
\begin{equation}
\frac{1}{\lim\limits_{N\to\infty}F(X_{s,t})}=\left[\tfrac{7}{6}+8\left(\abs{s}-\tfrac{1}{4}\right)^2\right]\left[\tfrac{7}{6}+8\left(\abs{t}-\tfrac{1}{4}\right)^2\right]-1.    \label{eqn:F_Legendre}
\end{equation}
\end{theorem}
\begin{proof}
Let $L$ and $K$ be the ternary Legendre sequences of length $p$ and $q$, respectively, and write $T:=L\times K$. Notice that $T_{s,t}=L_s\times K_t$ and $T_{s,t}(x,y)=L_s(x)K_t(y)$. Then from~\eqref{eqn:C_product}
\[
\sum_{u,v}[C_{T_{s,t}}(u,v)]^2=\sum_u[C_{L_s}(u)]^2\cdot \sum_v[C_{K_t}(v)]^2.
\]
The condition~\eqref{eqn:cond_pq} implies that $p$ and $q$ grow without bound as $N\to\infty$. We can therefore apply Proposition~\ref{prop:FLegendre} to give
\begin{equation}
\frac{1}{(pq)^2}\sum_{u,v}[C_{T_{s,t}}(u,v)]^2=\left[\tfrac{7}{6}+8\left(\abs{s}-\tfrac{1}{4}\right)^2\right]\left[\tfrac{7}{6}+8\left(\abs{t}-\tfrac{1}{4}\right)^2\right]+o(1)\quad\mbox{as $N\to\infty$}.   \label{eqn:ss_ternary_L_array}
\end{equation}
Define
\begin{equation}
\Delta(N):=\frac{1}{(pq)^2}\Biggabs{\sum_{u,v}[C_{X_{s,t}}(u,v)]^2-\sum_{u,v}[C_{T_{s,t}}(u,v)]^2}.   \label{eqn:def_Delta}
\end{equation}
We claim that
\begin{equation}
\Delta(N)\to 0\quad\mbox{as $N\to\infty$}.   \label{eqn:claim_Delta}
\end{equation}
The theorem then follows from~\eqref{eqn:ss_ternary_L_array} and~\eqref{eqn:def_Delta} and the fact $C_{X_{s,t}}(0,0)=pq$.
\par
It remains to prove the claim~\eqref{eqn:claim_Delta}. By the definition of a Legendre array, there exists an array $V\in\V_{p,q}$ such that $X_{s,t}=T_{s,t}+V_{s,t}$. From Lemma~\ref{lem:CQ} we then have
\begin{align}
\Delta(N)&=\frac{1}{4(pq)^3}\Biggabs{\sum_{i=0}^{2p-1}\sum_{j=0}^{2q-1}\bigabs{T_{s,t}(\z_p^i,\zeta_q^j)+V_{s,t}(\z_p^i,\zeta_q^j)}^4-\sum_{i=0}^{2p-1}\sum_{j=0}^{2q-1}\bigabs{T_{s,t}(\z_p^i,\zeta_q^j)}^4}   \nonumber\\
&\le\frac{1}{4(pq)^3}\sum_{i=0}^{2p-1}\sum_{j=0}^{2q-1}\biggabs{\bigabs{T_{s,t}(\z_p^i,\zeta_q^j)+V_{s,t}(\z_p^i,\zeta_q^j)}^4-\bigabs{T_{s,t}(\z_p^i,\zeta_q^j)}^4}.   \label{eqn:Delta_bound}
\end{align}
Now for $a,b\in\C$ the identity
\[
\abs{a+b}^4=\abs{a}^4+\abs{b}^4+4\big[\Re(a\overline{b})\big]^2+2\,\abs{a}^2\cdot\abs{b}^2+4\,\abs{a}^2\cdot\Re(a\overline{b})+4\,\abs{b}^2\cdot\Re(a\overline{b})
\]
gives the inequality
\[
\Bigabs{\abs{a+b}^4-\abs{a}^4}\le 4\,\abs{a}^3\cdot \abs{b}+6\,\abs{a}^2\cdot\abs{b}^2+4\,\abs{a}\cdot\abs{b}^3+\abs{b}^4.
\]
Apply this bound to~\eqref{eqn:Delta_bound} to obtain
\begin{align}
\Delta(N)&\le \frac{1}{(pq)^3}\sum_{i=0}^{2p-1}\sum_{j=0}^{2q-1}\bigabs{T_{s,t}(\z_p^i,\zeta_q^j)}^3\cdot\bigabs{V_{s,t}(\z_p^i,\zeta_q^j)}   \nonumber\\
&+\frac{3}{2(pq)^3}\sum_{i=0}^{2p-1}\sum_{j=0}^{2q-1}\bigabs{T_{s,t}(\z_p^i,\zeta_q^j)}^2\cdot\bigabs{V_{s,t}(\z_p^i,\zeta_q^j)}^2   \nonumber\\
&+\frac{1}{(pq)^3}\sum_{i=0}^{2p-1}\sum_{j=0}^{2q-1}\bigabs{T_{s,t}(\z_p^i,\zeta_q^j)}\cdot\bigabs{V_{s,t}(\z_p^i,\zeta_q^j)}^3   \nonumber\\
&+\frac{1}{4(pq)^3}\sum_{i=0}^{2p-1}\sum_{j=0}^{2q-1}\bigabs{V_{s,t}(\z_p^i,\zeta_q^j)}^4.   \label{eqn_Delta_bound2}
\end{align}
Given a ternary Legendre sequence $A$ of length $d$, it is well known (see~\cite[p.~182]{Schroeder1997}, for example) that $\abs{A(\z_d^{2k})}\le d^{1/2}$ for each integer $k$. It is easily verified that this implies $\abs{A_r(\z_d^{2k})}\le d^{1/2}$ for each integer $k$ and all real $r$. Therefore, since $T_{s,t}(x,y)=L_s(x)K_t(y)$, Lemma~\ref{lem:poly-bound} gives
\[
\bigabs{T_{s,t}(\z_p^i,\zeta_q^j)}\le 4(pq)^{1/2}\,\log(p+q)\quad\mbox{for all integers $i$ and $j$}.
\]
Substitute into~\eqref{eqn_Delta_bound2} to give
\begin{equation}
\Delta(N)\le 256\,\frac{(\log(p+q))^3}{(pq)^{1/2}}S_1+96\,\frac{(\log(p+q))^2}{pq}S_2+16\,\frac{\log(p+q)}{(pq)^{3/2}}S_3+\frac{1}{(pq)^2}S_4,   \label{eqn:Delta_bound3}
\end{equation}
where
\[
S_\ell:=\frac{1}{4pq}\sum_{i=0}^{2p-1}\sum_{j=0}^{2q-1}\bigabs{V_{s,t}(\z_p^i,\z_q^j)}^\ell.
\]
From~\eqref{eqn:abs_A_poly} and a straightforward calculation we obtain
\begin{align*}
S_2&=C_{V_{s,t}}(0,0)\\
&=p+q-1.
\end{align*}
The Cauchy--Schwarz inequality gives $S_1\le (S_2)^{1/2}$, and since $\abs{V_{s,t}(x,y)}\le p+q-1$ for $\abs{x}=\abs{y}=1$, we also have $S_\ell\le (p+q-1)^{\ell-2}S_2$ for $\ell\ge 2$. Substitution into~\eqref{eqn:Delta_bound3} gives
\begin{align*}
\Delta(N)&\le 256\,\left(\frac{p+q-1}{pq}\right)^{1/2}(\log(p+q))^3+96\,\frac{p+q-1}{pq}(\log(p+q))^2\\
&\quad+16\,\frac{(p+q-1)^2}{(pq)^{3/2}}\log(p+q)+\frac{(p+q-1)^3}{(pq)^2}.
\end{align*}
Now, using the condition~\eqref{eqn:cond_pq}, we readily verify our claim~\eqref{eqn:claim_Delta}.
\end{proof}
\par
There is no loss of generality in Theorem~\ref{thm:Farray1} from the restriction $-\tfrac{1}{2}<s,t\le\tfrac{1}{2}$ since $A_{s,t}$ is the same as $A_{s+1,t}$ and $A_{s,t+1}$ for every array $A$ and all real $s$ and $t$. We note that the condition~\eqref{eqn:cond_pq} can be relaxed for particular Legendre arrays. For example, let $L$ and $K$ be the Legendre sequences of length $p$ and $q$, respectively. Then $X=L\times K$ is a Legendre array. From~\eqref{eqn:F_prod} and Theorem~\ref{thm:hj} we conclude that~\eqref{eqn:F_Legendre} holds under the relaxed condition $p\to\infty$ and $q\to\infty$ as $N\to\infty$.


\section{The Merit Factor of Quadratic-Residue Arrays}

In this section our goal is to calculate the asymptotic merit factor of a quadratic-residue array of size $p\times p$ at all rotations. We shall assume throughout this section that $p$ is an odd prime. Write the $p$th roots of unity as
\[
\e_j:=e^{\sqrt{-1}\,2\pi j/p}\quad\mbox{for integer $j$}.
\]
Then, since $p$ is odd, we have
\[
\{\zeta_p^i:0\le i<2p\}=\{\e_j:0\le j<p\}\cup\{-\e_j:0\le j<p\}.
\]
Therefore, given an array $A$ of size $p\times p$, Lemma~\ref{lem:CQ} asserts that
\begin{multline}
\sum_{u,v}[C_A(u,v)]^2\\
=\frac{1}{4p^2}\sum_{0\le i,j<p}\left(\abs{A(\e_i,\e_j)}^4+\abs{A(-\e_i,\e_j)}^4+\abs{A(\e_i,-\e_j)}^4+\abs{A(-\e_i,-\e_j)}^4\right).   \label{eqn:CQ}
\end{multline}
Our objective is to find an asymptotic expression for the sum on the right-hand side of the identity~\eqref{eqn:CQ}, where $A$ is a rotated ternary quadratic-residue array. Since a ternary quadratic-residue array and a quadratic-residue array differ in only one element, this will be sufficient to compute the asymptotic merit factor of a rotated quadratic-residue array. Before we analyse the sum in~\eqref{eqn:CQ}, we shall need several technical results, which we state in the next subsection.

\subsection{Auxiliary Results}

The following lemma evaluates the generating function of a ternary quadratic-residue array at $p$th roots of unity.
\begin{lemma}
\label{lem:gauss_sum}
Let $\chi$ be the quadratic character of $\GF(p^2)$, and let $Z$ be a ternary quadratic-residue array of size $p\times p$, as defined in~\eqref{eqn:QRarray2}. Then there exists a basis $\{\beta,\beta'\}$ for $\GF(p^2)$ over $\GF(p)$ such that
\[
Z(\e_k,\e_\ell)=(-1)^{\frac{p+1}{2}}p\,\chi(k\beta+\ell\beta')
\quad\mbox{for all integers $k,\ell$}.
\]
\end{lemma}
\begin{proof}
Let $\Tr:\GF(p^2)\to\GF(p)$ be the trace function given by
\[
\Tr(x)=x+x^p,
\]
and for $b\in\GF(p^2)$, let $\psi_b:\GF(p^2)\to\C$ be the additive character of $\GF(p^2)$ given by
\[
\psi_b(x):=e^{\sqrt{-1}\,2\pi\Tr(bx)/p}.
\]
It is readily verified that
\begin{equation}
\label{eqn:trace_properties}
\Tr(ax+by)=a\Tr(x)+b\Tr(y)\quad\mbox{for $a,b\in \GF(p)$.}
\end{equation}
We choose $\{\beta,\beta'\}$ such that $\{\alpha,\alpha'\}$ (appearing in the definition~\eqref{eqn:QRarray2} of $Z$) and $\{\beta,\beta'\}$ are dual bases, that is,
\begin{equation}
\label{eqn:dual_basis}
\Tr(\alpha\beta)=1,\quad\Tr(\alpha\beta')=0,\quad\Tr(\alpha'\beta)=0,\quad\Tr(\alpha'\beta')=1.
\end{equation}
Such a basis is guaranteed to exist~\cite[p.~58]{Lidl1997}. Given integers $i,j,k,\ell$, we then have by~\eqref{eqn:trace_properties} and~\eqref{eqn:dual_basis}
\[
\Tr((i\alpha+j\alpha')(k\beta+\ell\beta'))=ik+j\ell,
\]
and therefore
\begin{align*}
Z(\e_k,\e_\ell)&=\sum_{0\le i,j<p}\chi(i\alpha+j\alpha')e^{\sqrt{-1}\,2\pi (ik+j\ell)/p}\\
&=\sum_{0\le i,j<p}\chi(i\alpha+j\alpha')e^{\sqrt{-1}\,2\pi \Tr((i\alpha+j\alpha')(k\beta+\ell\beta'))/p}\\
&=\sum_{a\in\GF(p^2)}\chi(a)\psi_{k\beta+\ell\beta'}(a)
\end{align*}
by putting $a:=i\alpha+j\alpha'$. The above sum is called a \emph{Gaussian sum}, and it is well known that
\[
\sum_{a\in\GF(p^2)}\chi(a)\psi_b(a)=(-1)^{\frac{p+1}{2}}p\,\chi(b)
\]
(see~\cite[pp.~199--201]{Lidl1997}, for example). This proves the lemma.
\end{proof}
\par
Our next lemma bounds a certain character sum and evaluates it in special cases.
\begin{lemma}
Let $\chi$ be the quadratic character of $\GF(p^2)$, and define
\begin{equation}
\Omega(\kappa,\lambda,\mu):=\sum_{x\in\GF(p^2)}\chi(x)\chi(x+\kappa)\chi(x+\lambda)\chi(x+\mu)\quad\mbox{for $\kappa,\lambda,\mu\in\GF(p^2)$}   \label{eqn:def_Omega}
\end{equation}
and
\begin{equation}
I:=\{(\kappa,\kappa,0):\kappa\in\GF(p^2)\}\cup\{(\kappa,0,\kappa):\kappa\in\GF(p^2)\}\cup\{(0,\kappa,\kappa):\kappa\in\GF(p^2)\}.   \label{eqn:def_I}
\end{equation}
Then
\begin{equation}
\Omega(\kappa,\kappa,0)=\begin{cases}
p^2-1 & \mbox{for $\kappa=0$}\\
p^2-2 & \mbox{for $\kappa\ne 0$}
\end{cases}   \label{eqn:Omegakk0}
\end{equation}
and
\begin{equation}
\bigabs{\Omega(\kappa,\lambda,\mu)}\le 3p\quad\mbox{for $(\kappa,\lambda,\mu)\not\in I$}.   \label{eqn:Omega_bound}
\end{equation}
\end{lemma}
\begin{proof}
Since $\chi$ is multiplicative by~\eqref{eqn:multiplicativity},
\begin{align*}
\Omega(\kappa,\kappa,0)&=\sum_{x\in\GF(p^2)}[\chi(x(x+\kappa))]^2\\
&=\sum_{x\in\GF(p^2)\setminus\{0,-\kappa\}}1
\end{align*}
using $\chi(0)=0$ and $[\chi(x)]^2=1$ for all nonzero $x\in\GF(p^2)$. This proves~\eqref{eqn:Omegakk0}.
\par
To prove~\eqref{eqn:Omega_bound}, we use a special case of a result on multiplicative character sums with polynomial arguments~\cite[Thm.~5.41]{Lidl1997}, which can be stated as follows. If $f\in\GF(p^2)[x]$ is a monic polynomial of positive degree $d$ that is not a square (that is, $f(x)$ cannot be written as $f(x)=[g(x)]^2$ for some polynomial $g\in\GF(p^2)[x]$), then
\begin{equation}
\Biggabs{\sum_{x\in\GF(p^2)}\chi(f(x))}\le (d-1)p.   \label{eqn:char_sum}
\end{equation}
For all $(\kappa,\lambda,\mu)\not\in I$, the polynomial $f(x):=x(x+\kappa)(x+\lambda)(x+\mu)$ is not a square. Using the multiplicativity~\eqref{eqn:multiplicativity} of $\chi$, application of~\eqref{eqn:char_sum} gives~\eqref{eqn:Omega_bound}.
\end{proof}
\par
The following technical lemma can be obtained from the results of~\cite{Hoholdt1988} and~\cite{Jensen1991}.
\begin{lemma}
\label{lem:part_frac_exp}
Define
\begin{equation}
\Gamma(k,\ell,m):=\sum_{i=0}^{p-1} \frac{\e_i^2}{(1+\e_i)(\e_k+\e_i)(\e_\ell+\e_i)(\e_m+\e_i)}\quad\mbox{for integer $k,\ell,m$}.   \label{eqn:def_Gamma}
\end{equation}
Then
\begin{equation}
\Gamma(k,k,0)= 
\begin{cases}
\dfrac{p^2(p^2+2)}{48}
			& \mbox{for $k\equiv 0\pmod p$} 				\\[2ex]
\dfrac{p^2}{2} \cdot \dfrac{1}{\e_k\,\abs{1-\e_k}^2}
			& \mbox{for $k\not\equiv 0\pmod p$}
\end{cases}   \label{eqn:Gammakk0}
\end{equation}
and
\begin{equation}
\sum_{0\le k,\ell,m<p}\bigabs{\Gamma(k,\ell,m)}\le (p\log p)^4.   \label{eqn:Gamma_bound}
\end{equation}
\end{lemma}
\begin{proof}
The identity~\eqref{eqn:Gammakk0} was established in~\cite[p.~162, Cases~4 and 5]{Hoholdt1988}. The bound~\eqref{eqn:Gamma_bound} follows from the inequality $\sum_{i=0}^{p-1}1/\abs{1+\e_i}\le p\log p$ (see~\cite[p.~625]{Jensen1991}, for example) since
\begin{align*}
\sum_{0\le k,\ell,m<p}\bigabs{\Gamma(k,\ell,m)}&\le \sum_{0<i,k,\ell,m<p}\biggabs{\frac{\e^2_i}{(1+\e_i)(\e_k+\e_i)(\e_\ell+\e_i)(\e_m+\e_i)}}\\
&=\left(\sum_{i=0}^{p-1}\frac{1}{\abs{1+\e_i}}\right)^4.   \qedhere
\end{align*}
\end{proof}


\subsection{Asymptotic Merit Factor Calculation}

We are now in a position to analyse asymptotic behaviour of the sum on the right-hand side of the identity~\eqref{eqn:CQ}, where $A$ is a rotated ternary quadratic-residue array. We split the analysis into the following three lemmas.
\begin{lemma}
\label{lem:Qsum0}
Let $Z$ be a ternary quadratic-residue array of size $p\times p$, and let $s$ and $t$ be real numbers satisfying $-\tfrac{1}{2}<s,t\le\tfrac{1}{2}$. Then, as $p\to\infty$,
\begin{align*}
\frac{1}{4p^6}\sum_{0\le i,j<p}\bigabs{Z_{s,t}(\e_i,\e_j)}^4&=\tfrac{1}{4}+O(p^{-2}).
\end{align*}
\end{lemma}
\begin{proof}
Lemma~\ref{lem:gauss_sum} implies
\[
\bigabs{Z(\e_i,\e_j)}=\begin{cases}
0 & \quad\mbox{for $i\equiv j\equiv 0\pmod p$}\\
p & \quad\mbox{otherwise}.
\end{cases}
\]
Then, using the easily verified identity
\begin{equation}
Z_{s,t}(\e_i,\e_j)=\e_i^{-\lfloor ps\rfloor} \e_j^{-\lfloor pt\rfloor} Z(\e_i,\e_j),   \label{eqn:Xst_X}
\end{equation}
we find that
\[
\frac{1}{4p^6}\sum_{0\le i,j<p}\bigabs{Z_{s,t}(\e_i,\e_j)}^4=\frac{1}{4}\left(1-\frac{1}{p^2}\right),
\]
as required.
\end{proof}
\par
\begin{lemma}
\label{lem:Qsum1}
Let $Z$ be a ternary quadratic-residue array of size $p\times p$, and let $s$ and $t$ be real numbers satisfying $-\tfrac{1}{2}<s,t\le\tfrac{1}{2}$. Then, as $p\to\infty$,
\begin{align*}
\frac{1}{4p^6}\sum_{0\le i,j<p}\bigabs{Z_{s,t}(-\e_i,\e_j)}^4&=\tfrac{1}{3}+4(\abs{s}-\tfrac{1}{4})^2+O\left(p^{-1}(\log p)^4\right),\\
\frac{1}{4p^6}\sum_{0\le i,j<p}\bigabs{Z_{s,t}(\e_i,-\e_j)}^4&=\tfrac{1}{3}+4(\abs{t}-\tfrac{1}{4})^2+O\left(p^{-1}(\log p)^4\right).
\end{align*}
\end{lemma}
\begin{proof}
Given an array $A$, let $A^T$ denote the transpose of $A$. Since $(Z_{s,t})^T=(Z^T)_{t,s}$ and $Z^T$ is again a ternary quadratic-residue array, it is sufficient to prove the first statement in the lemma.
\par
Let $P\in\C[x]$ be a polynomial of degree $p-1$, and let $i$ and $j$ be integer. We shall make use of the Lagrange interpolation formula
\begin{align}
P(-\e_i) & = \frac{2}{p} \sum_{k=0}^{p-1} P(\e_k) \frac{\e_k}{\e_k+\e_i}   \label{eqn:interpolation}\\
\intertext{(see \cite[Eq.~(2.5)]{Hoholdt1988}, for example). It follows that}
Z_{s,t}(-\e_i,\e_j) & = \frac{2}{p}\sum_{k=0}^{p-1} Z_{s,t}(\e_k,\e_j) \frac{\e_k}{\e_k+\e_i}.   \nonumber
\intertext{Set $S:=\lfloor ps\rfloor$ and $T:=\lfloor pt\rfloor$, and use~\eqref{eqn:Xst_X} and Lemma~\ref{lem:gauss_sum} to obtain}
Z_{s,t}(-\e_i,\e_j) & = 2(-1)^{\frac{p+1}{2}} \e_j^{-T} \sum_{k=0}^{p-1} 
\e_k^{-S} \chi(k\beta+j\beta') \frac{\e_k}{\e_k+\e_i},    \nonumber
\intertext{where $\chi$ is the quadratic character of $\GF(p^2)$ and $\{\beta,\beta'\}$ is some basis for $\GF(p^2)$ over $\GF(p)$. Since we also have}
\overline{Z_{s,t}(-\e_i,\e_j)} & = 
2(-1)^{\frac{p+1}{2}} \e_j^T \sum_{k=0}^{p-1} \e_k^S 
\chi(k\beta+j\beta') \frac{\e_i}{\e_k+\e_i}, 		\nonumber
\end{align}
we find that
\begin{multline*}
\bigabs{Z_{s,t}(-\e_i,\e_j)}^4=16\!\!\!\!\sum_{0\le a,b,c,d<p}\!\!\!\!\e_{b-a}^S\e_{d-c}^S\,\chi(a\beta+j\beta')\chi(b\beta+j\beta')\chi(c\beta+j\beta')\chi(d\beta+j\beta')\\
\frac{\e_a}{\e_a+\e_i}\frac{\e_i}{\e_b+\e_i}\frac{\e_c}{\e_c+\e_i}\frac{\e_i}{\e_d+\e_i}.
\end{multline*}
Use the definition~\eqref{eqn:def_Gamma} of $\Gamma(k,\ell,m)$ to write
\begin{equation}
\sum_{i=0}^{p-1}\frac{\e_a}{\e_a+\e_i}\frac{\e_i}{\e_b+\e_i}\frac{\e_c}{\e_c+\e_i}\frac{\e_i}{\e_d+\e_i}=\e_{c-a}\,\Gamma(b-a,c-a,d-a),   \label{eqn:Gamma_applied}
\end{equation}
so that after variable relabeling
\begin{multline*}
\sum_{i=0}^{p-1}\bigabs{Z_{s,t}(-\e_i,\e_j)}^4=16\sum_{0\le a,k,\ell,m<p}\e_{k-\ell+m}^S\e_\ell\\
\chi(a\beta+j\beta')\chi(a\beta+j\beta'+k\beta)\chi(a\beta+j\beta'+\ell\beta)\chi(a\beta+j\beta'+m\beta)\,\Gamma(k,\ell,m).
\end{multline*}
Put $x:=a\beta+j\beta'$ and note that $a\beta+j\beta'$ ranges over $\GF(p^2)$ as $a$ and $j$ range from $0$ to $p-1$. By the definition~\eqref{eqn:def_Omega} of $\Omega(\kappa,\lambda,\mu)$ we therefore obtain
\[
\sum_{0\le i,j<p}\bigabs{Z_{s,t}(-\e_i,\e_j)}^4=16\sum_{0\le k,\ell,m<p}\e_{k-\ell+m}^S\e_\ell\;\Omega(k\beta,\ell\beta,m\beta)\,\Gamma(k,\ell,m),
\]
Let the set $I$ be as defined in~\eqref{eqn:def_I}, and write
\begin{equation}
\frac{1}{4p^6}\sum_{0\le i,j<p}\bigabs{Z_{s,t}(-\e_i,\e_j)}^4=A+B,   \label{eqn:split_sum_1}
\end{equation}
where
\begin{align*}
A&=\frac{4}{p^6}\sum_{\stack{0\le k,\ell,m<p}{(k\beta,\ell\beta,m\beta)\not\in I}}\e_{k-\ell+m}^S\e_\ell\;\Omega(k\beta,\ell\beta,m\beta)\,\Gamma(k,\ell,m)\\
B&=\frac{4}{p^6}\sum_{\stack{0\le k,\ell,m<p}{(k\beta,\ell\beta,m\beta)\in I}}\e_{k-\ell+m}^S\e_\ell\;\Omega(k\beta,\ell\beta,m\beta)\,\Gamma(k,\ell,m).
\end{align*}
Using~\eqref{eqn:Omega_bound}, the sum $A$ can be bounded as
\begin{align}
\bigabs{A}
&\le \frac{12}{p^5}\sum_{0\le k,\ell,m<p}\bigabs{\Gamma(k,\ell,m)}   \nonumber\\
&\le12\,p^{-1}(\log p)^4   \nonumber
\intertext{by~\eqref{eqn:Gamma_bound}. Therefore,}
A&=O(p^{-1}(\log p)^4)\quad\mbox{as $p\to\infty$}.   \label{eqn:bound_S1}
\end{align}
Using symmetry of $\Omega(\kappa,\lambda,\mu)$ and $\Gamma(k,\ell,m)$ with respect to interchanging their arguments, the sum $B$ can be written as
\[
B=\frac{4}{p^6}\,\Omega(0,0,0)\,\Gamma(0,0,0)+\frac{4}{p^6}\,\sum_{k=1}^{p-1}\Omega(k\beta,k\beta,0)\Gamma(k,k,0)\left(2\e_k+\e_k^{2S}\right).
\]
Application of~\eqref{eqn:Omegakk0} and~\eqref{eqn:Gammakk0} to evaluate $\Omega(k\beta,k\beta,0)$ and $\Gamma(k,k,0)$ then gives
\begin{equation}
B=\frac{(p^2+2)(p^2-1)}{12p^4}+\frac{2(p^2-2)}{p^4}\sum_{k=1}^{p-1}\frac{2+\e_k^{2S-1}}{\abs{1-\e_k}^2}.   \label{eqn:sumQ_subst}
\end{equation}
We wish to apply the identity 
\begin{equation}
\sum_{k=1}^{p-1} \frac{\e_k^j}{\abs{1-\e_k}^2}
= \frac{p^2}{2} \left (\frac{\abs{j}}{p} - \frac{1}{2} \right )^2 
					 - \frac{p^2+2}{24}
			\quad \mbox{for integer $j$ satisfying $\abs{j} \le p$}
\label{eqn:exp_sum_identity}
\end{equation}
(see \cite[p.~621]{Jensen1991}, for example). Since $\tfrac{1}{2}<s\le \tfrac{1}{2}$, we have $\abs{2S-1}\le p$ for all sufficiently large~$p$, which allows us to apply~\eqref{eqn:exp_sum_identity} to~\eqref{eqn:sumQ_subst} to obtain, for all sufficiently large~$p$,
\[
B=\frac{(p^2+2)(p^2-1)}{12p^4}+\frac{(p^2-2)^2}{4p^4}+\frac{p^2-2}{p^2}\left(\frac{\abs{2S-1}}{p}-\frac{1}{2}\right)^2.
\]
By the definition of $S$, we have $S=ps+O(1)$ as $p\to\infty$, and therefore,
\begin{equation}
B=\tfrac{1}{3}+4\left(\abs{s}-\tfrac{1}{4}\right)^2+O(p^{-1})\quad\mbox{as $p\to\infty$}.   \label{eqn:S2}
\end{equation}
The proof is completed by substituting~\eqref{eqn:bound_S1} and~\eqref{eqn:S2} in~\eqref{eqn:split_sum_1}.
\end{proof}
\par
\begin{lemma}
\label{lem:Qsum2}
Let $Z$ be a ternary quadratic-residue array of size $p\times p$, and let $s$ and $t$ be real numbers satisfying $-\tfrac{1}{2}<s,t\le\tfrac{1}{2}$. Then, as $p\to\infty$,
\[
\frac{1}{4p^6}\sum_{0\le i,j<p}\bigabs{Z_{s,t}(-\e_i,-\e_j)}^4=\tfrac{4}{9}+64(\abs{s}-\tfrac{1}{4})^2(\abs{t}-\tfrac{1}{4})^2+O\left(p^{-1}(\log p)^8\right).
\]
\end{lemma}
\begin{proof}
The idea of the proof is similar to that of the proof of Lemma~\ref{lem:Qsum1}. The main difference is that we now have to apply interpolation of $Z_{s,t}(x,y)$ in both indeterminates.
\par
Let $i$ and $j$ be integer. By the interpolation formula~\eqref{eqn:interpolation} we have
\begin{align*}
Z_{s,t}(-\e_i,-\e_j) & = \frac{2}{p}\sum_{k=0}^{p-1} Z_{s,t}(\e_k,-\e_j) \frac{\e_k}{\e_k+\e_i}.
\intertext{Apply the interpolation formula~\eqref{eqn:interpolation} again to obtain}
Z_{s,t}(-\e_i,-\e_j) & = \frac{4}{p^2}\sum_{0\le k,\ell<p} Z_{s,t}(\e_k,\e_\ell) \frac{\e_k}{\e_k+\e_i}\frac{\e_\ell}{\e_\ell+\e_j}.
\intertext{Set $S:=\lfloor ps\rfloor$ and $T:=\lfloor pt\rfloor$. Then by~\eqref{eqn:Xst_X} and Lemma~\ref{lem:gauss_sum} we get}
Z_{s,t}(-\e_i,-\e_j) & = \frac{4}{p}(-1)^{\frac{p+1}{2}}\sum_{0\le k,\ell<p}\e_k^{-S}\e_\ell^{-T}\chi(k\beta+\ell\beta') \frac{\e_k}{\e_k+\e_i}\frac{\e_\ell}{\e_\ell+\e_j},
\intertext{where $\chi$ is the quadratic character of $\GF(p^2)$ and $\{\beta,\beta'\}$ is some basis for $\GF(p^2)$ over $\GF(p)$. Since we also have}
\overline{Z_{s,t}(-\e_i,-\e_j)} & = \frac{4}{p}(-1)^{\frac{p+1}{2}}\sum_{0\le k,\ell<p}\e_k^S\e_\ell^T\chi(k\beta+\ell\beta') \frac{\e_i}{\e_k+\e_i}\frac{\e_j}{\e_\ell+\e_j},
\end{align*}
we obtain
\begin{multline*}
\bigabs{Z_{s,t}(-\e_i,-\e_j)}^4=\left(\frac{4}{p}\right)^4\sum_{0\le a,b,c,d<p}\,\sum_{0\le a',b',c',d'<p}\e_{b-a+d-c}^S\,\e_{b'-a'+d'-c'}^T\\
\chi(a\beta+a'\beta')\chi(b\beta+b'\beta')\chi(c\beta+c'\beta')\chi(d\beta+d'\beta')\\
\frac{\e_a}{\e_a+\e_i}\frac{\e_{a'}}{\e_{a'}+\e_j}\;\frac{\e_i}{\e_b+\e_i}\frac{\e_j}{\e_{b'}+\e_j}\;\frac{\e_c}{\e_c+\e_i}\frac{\e_{c'}}{\e_{c'}+\e_j}\;\frac{\e_i}{\e_d+\e_i}\frac{\e_j}{\e_{d'}+\e_j}.
\end{multline*}
Use~\eqref{eqn:Gamma_applied}, relabel the summation indices, and use the definition~\eqref{eqn:def_Omega} of $\Omega(\kappa,\lambda,\mu)$ to give
\begin{multline*}
\sum_{0\le i,j<p}\bigabs{Z_{s,t}(-\e_i,-\e_j)}^4=\frac{256}{p^4}\sum_{0\le k,\ell,m<p}\,\sum_{0\le k',\ell',m'<p}\e_{k-\ell+m}^S\,\e_{k'-\ell'+m'}^T\,\e_\ell\,\e_{\ell'}\\
\Omega(k\beta+k'\beta',\ell\beta+\ell'\beta',m\beta+m'\beta')\,\Gamma(k,\ell,m)\,\Gamma(k',\ell',m').
\end{multline*}
Let $I$ be the set defined in~\eqref{eqn:def_I}, and write
\begin{equation}
\frac{1}{4p^6}\sum_{0\le i,j<p}\bigabs{Z_{s,t}(-\e_i,-\e_j)}^4=A+B,   \label{eqn:sum_S_AB}
\end{equation}
where
\begin{align*}
A&=\frac{64}{p^{10}}\sum_{\stack{0\le k,\ell,m,k',\ell',m'<p}{(k\beta+k'\beta',\ell\beta+\ell'\beta',m\beta+m'\beta')\not\in I}}\e_{k-\ell+m}^S\,\e_{k'-\ell'+m'}^T\,\e_\ell\,\e_{\ell'}\\[1ex]
&\hspace{5em}\Omega(k\beta+k'\beta',\ell\beta+\ell'\beta',m\beta+m'\beta')\,\Gamma(k,\ell,m)\,\Gamma(k',\ell',m')\\[1ex]
B&=\frac{64}{p^{10}}\sum_{\stack{0\le k,\ell,m,k',\ell',m'<p}{(k\beta+k'\beta',\ell\beta+\ell'\beta',m\beta+m'\beta')\in I}}\e_{k-\ell+m}^S\,\e_{k'-\ell'+m'}^T\,\e_\ell\,\e_{\ell'}\\[1ex]
&\hspace{5em}\Omega(k\beta+k'\beta',\ell\beta+\ell'\beta',m\beta+m'\beta')\,\Gamma(k,\ell,m)\,\Gamma(k',\ell',m').
\end{align*}
From~\eqref{eqn:Omega_bound} we have the bound
\begin{align}
\bigabs{A}&\le \frac{192}{p^9}\Bigg(\sum_{0\le k,\ell,m<p}\bigabs{\Gamma(k,\ell,m)}\Bigg)^2   \nonumber\\
&\le192\,p^{-1}(\log p)^8   \nonumber
\intertext{by~\eqref{eqn:Gamma_bound}, giving}
A&=O(p^{-1}(\log p)^8)\quad\mbox{as $p\to\infty$}.   \label{eqn:bound_S1_2}
\end{align}
We use symmetry of $\Omega(\kappa,\lambda,\mu)$ and $\Gamma(k,\ell,m)$ with respect to interchanging their arguments to partition the sum $B$ further as
\begin{equation}
\label{eqn:sum:ABCD}
B=B_1+B_2+B_3+B_4,
\end{equation}
where
\begin{align*}
B_1&=\frac{64}{p^{10}}\,\Omega(0,0,0)\,[\Gamma(0,0,0)]^2\\
B_2&=\frac{64}{p^{10}}\sum_{k=1}^{p-1}\Omega(k\beta,k\beta,0)\,\Gamma(k,k,0)\,\Gamma(0,0,0)(2\e_k+\e_k^{2S})\\
B_3&=\frac{64}{p^{10}}\sum_{k'=1}^{p-1}\Omega(k'\beta',k'\beta',0)\,\Gamma(0,0,0)\,\Gamma(k',k',0)(2\e_{k'}+\e_{k'}^{2T})\\
B_4&=\frac{64}{p^{10}}\sum_{1\le k,k'<p}\Omega(k\beta+k'\beta',k\beta+k'\beta',0)\,\Gamma(k,k,0)\,\Gamma(k',k',0)(2\e_k\e_{k'}+\e_k^{2S}\e_{k'}^{2T}).
\end{align*}
Application of~\eqref{eqn:Omegakk0} and~\eqref{eqn:Gammakk0} to evaluate $\Omega(\kappa,\kappa,0)$ and $\Gamma(k,k,0)$ then gives
\begin{align*}
B_1&=\frac{1}{36}\;\frac{(p^2+2)^2(p^2-1)}{p^6}\\
B_2&=\frac{2}{3}\;\frac{p^4-4}{p^6}\sum_{k=1}^{p-1}\dfrac{2+\e_k^{2S-1}}{\abs{1-\e_k}^2}\\
B_3&=\frac{2}{3}\;\frac{p^4-4}{p^6}\sum_{k=1}^{p-1}\dfrac{2+\e_k^{2T-1}}{\abs{1-\e_k}^2}\\
B_4&=16\frac{p^2-2}{p^6}
\sum_{1\le k,k'<p}\frac{2+\e_k^{2S-1}\e_{k'}^{2T-1}}{\abs{1-\e_k}^2\;\abs{1-\e_{k'}}^2}.
\end{align*}
We complete the proof by evaluating the asymptotic behavior of the terms $B_1$, $B_2$, $B_3$, and $B_4$.
\begin{description}
\item[{\normalfont\itshape The term $B_1$}]
The term $B_1$ becomes
\begin{equation}
\label{eqn:term_A}
B_1=\frac{1}{36}+O(p^{-2})\quad\mbox{as $p\to\infty$}.
\end{equation}

\item[{\normalfont\itshape The terms $B_2$ and $B_3$}]
Since $-\tfrac{1}{2}<s\le \tfrac{1}{2}$ implies that $\abs{2S-1}\le p$ for all sufficiently large $p$, we can use the identity~\eqref{eqn:exp_sum_identity} to write
\begin{align}
B_2&=\frac{2}{3}\;\frac{p^4-4}{p^6}\left[\frac{p^2-2}{8}+\frac{p^2}{2}\left(\frac{\abs{2S-1}}{p}-\frac{1}{2}\right)^2\right].   \nonumber
\intertext{Then, since $S=ps+O(1)$ as $p\to\infty$, we obtain}
B_2&=\tfrac{1}{12}+\tfrac{4}{3}(\abs{s}-\tfrac{1}{4})^2+O(p^{-1})\quad\mbox{as $p\to\infty$}.   \label{eqn:term_B}
\intertext{We proceed similarly for the term $B_3$ and find that}
B_3&=\tfrac{1}{12}+\tfrac{4}{3}(\abs{t}-\tfrac{1}{4})^2+O(p^{-1})\quad\mbox{as $p\to\infty$}.   \label{eqn:term_C}
\end{align}

\item[{\normalfont\itshape The term $B_4$}] The term $B_4$ can be rewritten as
\[
B_4=16\frac{p^2-2}{p^6}\left[2\left(\sum_{k=1}^{p-1}\frac{1}{\abs{1-\e_k}^2}\right)^2+\left(\sum_{k=1}^{p-1}\frac{\e_k^{2S-1}}{\abs{1-\e_k}^2}\right)\left(\sum_{k'=1}^{p-1}\frac{\e_{k'}^{2T-1}}{\abs{1-\e_{k'}}^2}\right)\right].
\]
Noting that $\abs{2S-1}\le p$ and $\abs{2T-1}\le p$ for all sufficiently large $p$, application of~\eqref{eqn:exp_sum_identity} gives
\begin{multline*}
B_4=\frac{2}{9}\;\frac{(p^2-2)(p^2-1)^2}{p^6}\\
+16\frac{p^2-2}{p^6}\left[\frac{p^2}{2}\left(\frac{\abs{2S-1}}{p}-\frac{1}{2}\right)^2-\frac{p^2+2}{24}\right]\left[\frac{p^2}{2}\left(\frac{\abs{2T-1}}{p}-\frac{1}{2}\right)^2-\frac{p^2+2}{24}\right].
\end{multline*}
Since $S=ps+O(1)$ and $T=pt+O(1)$ as $p\to\infty$, we finally obtain
\begin{equation}
B_4=\tfrac{2}{9}+\left[8(\abs{s}-\tfrac{1}{4})^2-\tfrac{1}{6}\right]\left[8(\abs{t}-\tfrac{1}{4})^2-\tfrac{1}{6}\right]+O(p^{-1})\quad\mbox{as $p\to\infty$}.   \label{eqn:term_D}
\end{equation}
\end{description}
The result now follows by substituting the asymptotic forms~\eqref{eqn:term_A},~\eqref{eqn:term_B},~\eqref{eqn:term_C}, and~\eqref{eqn:term_D} into~\eqref{eqn:sum:ABCD} and then~\eqref{eqn:sum:ABCD} and~\eqref{eqn:bound_S1_2} into~\eqref{eqn:sum_S_AB}.
\end{proof}
We are now able to prove the main result of this section.
\begin{theorem}
Let $Y$ be a quadratic-residue array of size $p\times p$, and let $s$ and $t$ be real numbers satisfying $-\tfrac{1}{2}<s,t\le\tfrac{1}{2}$. Then
\[
\frac{1}{\lim\limits_{p\to\infty}F(Y_{s,t})}=\tfrac{1}{9}+\left[\tfrac{1}{2}+8\left(\abs{s}-\tfrac{1}{4}\right)^2\right]\left[\tfrac{1}{2}+8\left(\abs{t}-\tfrac{1}{4}\right)^2\right].
\]
\end{theorem}
\begin{proof}
Let $Z$ be a ternary quadratic-residue array of size $p\times p$. From~\eqref{eqn:CQ} we have
\begin{align*}
\frac{1}{p^4}\sum_{u,v}[C_{Z_{s,t}}(u,v)]^2&=\frac{1}{4p^6}\sum_{0\le i,j<p}\bigabs{Z_{s,t}(\e_i,\e_j)}^4+\frac{1}{4p^6}\sum_{0\le i,j<p}\bigabs{Z_{s,t}(-\e_i,\e_j)}^4\\
&+\frac{1}{4p^6}\sum_{0\le i,j<p}\bigabs{Z_{s,t}(\e_i,-\e_j)}^4+\frac{1}{4p^6}\sum_{0\le i,j<p}\bigabs{Z_{s,t}(-\e_i,-\e_j)}^4.
\end{align*}
Using Lemmas~\ref{lem:Qsum0},~\ref{lem:Qsum1}, and~\ref{lem:Qsum2} to determine the asymptotic behavior of the sums on the right-hand side, we obtain
\begin{multline}
\frac{1}{p^4}\sum_{u,v}[C_{Z_{s,t}}(u,v)]^2\\
=\tfrac{10}{9}+\left[\tfrac{1}{2}+8\left(\abs{s}-\tfrac{1}{4}\right)^2\right]\left[\tfrac{1}{2}+8\left(\abs{t}-\tfrac{1}{4}\right)^2\right]+O\left(p^{-1}(\log p)^8\right)\quad\mbox{as $p\to\infty$}.   \label{eqn:ss_Z}
\end{multline}
Now, since the binary array $Y$ differs from the ternary array $Z$ only in a single position, we have
\[
C_{Y_{s,t}}(u,v)=C_{Z_{s,t}}(u,v)+\delta(u,v),
\]
where
\begin{equation}
\abs{\delta(u,v)}\le\begin{cases}
1 & \mbox{for $-p<u,v<p$}\\
0 & \mbox{otherwise}.
\end{cases}   \label{eqn:abs_delta}
\end{equation}
Then, by the Cauchy--Schwarz inequality,
\begin{multline*}
\Biggabs{\sum_{u,v}[C_{Y_{s,t}}(u,v)]^2-\sum_{u,v}[C_{Z_{s,t}}(u,v)]^2}\\
\le\sum_{u,v}[\delta(u,v)]^2+2\bigg(\sum_{u,v}[C_{Z_{s,t}}(u,v)]^2\bigg)^{1/2}\bigg(\sum_{u,v}[\delta(u,v)]^2\bigg)^{1/2},
\end{multline*}
so that by~\eqref{eqn:abs_delta} and~\eqref{eqn:ss_Z}
\[
\frac{1}{p^4}\Biggabs{\sum_{u,v}[C_{Y_{s,t}}(u,v)]^2-\sum_{u,v}[C_{Z_{s,t}}(u,v)]^2}=O(p^{-1})\quad\mbox{as $p\to\infty$}.
\]
The theorem follows from~\eqref{eqn:ss_Z} by noting that $C_{Y_{s,t}}(0,0)=p^2$.
\end{proof}
\par


\section{Concluding Remarks}
\label{sec:conclusion}

We have computed the asymptotic value of the merit factor of Legendre arrays (under certain conditions on the growth rate of their dimensions) and of quadratic-residue arrays, for all rotations of rows and columns. The asymptotic merit factor of rotated Legendre arrays and rotated quadratic-residue arrays is shown in Figures~\ref{fig:FX} and~\ref{fig:FZ}, respectively. The maximum asymptotic merit factor, taken over all rotations, equals $\frac{36}{13}\simeq 2.77$ for both array families, which occurs at the rotations $(s,t)$, where $s,t\in\{\frac{1}{4},\frac{3}{4}\}$. However, at all other rotations, the asymptotic merit factor of quadratic-residue arrays is larger than that of Legendre arrays. On the other hand, an advantage of Legendre arrays is that they are not restricted to be square.
\begin{figure}[htp]
\centering
\includegraphics[width=12cm]{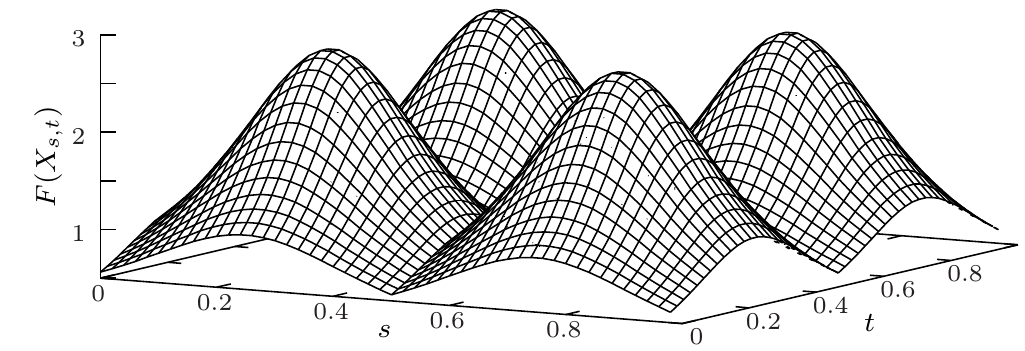}
\caption{The asymptotic value of $F(X_{s,t})$ (under conditions specified in Theorem~\ref{thm:Farray1}), where $X$ is a Legendre array.}
\label{fig:FX}
\end{figure}
\begin{figure}[htp]
\centering
\includegraphics[width=12cm]{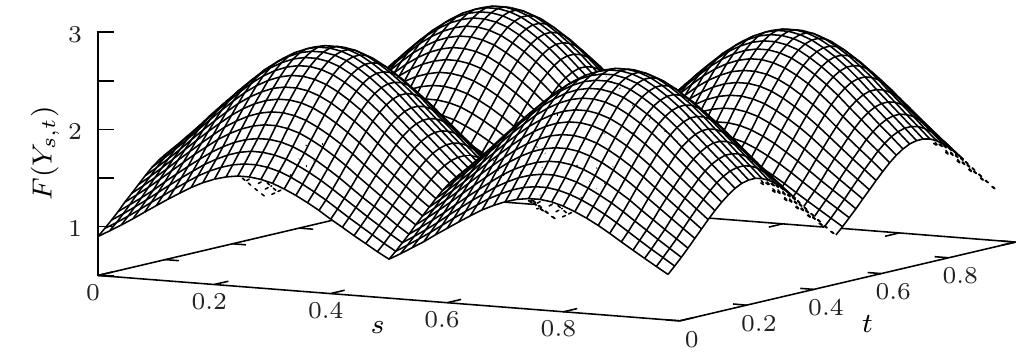}
\caption{The asymptotic value of $F(Y_{s,t})$, where $Y$ is a quadratic-residue array.}
\label{fig:FZ}
\end{figure}
\par
In~\cite{Borwein2004}, the authors exhibited a family of sequences $A$ obtained by appending an initial fraction of a rotated Legendre sequence to itself. Based on partial explanations and extensive numerical computations, it was conjectured in~\cite{Borwein2004} that this sequence family has asymptotic merit factor greater than~$6.34$. Under the assumption that this conjecture is correct, the corresponding square product array $A\times A$ has asymptotic merit factor greater than $2.93$, by~\eqref{eqn:F_prod}. This suggests that the maximum asymptotic value of the merit factor of the two array families considered in this paper, namely $\frac{36}{13}\simeq 2.77$, can be surpassed by another array family.
\par
In closing, we remark that instead of studying the merit factor of two-dimensional arrays of size $n\times m$, Gulliver and Parker~\cite{Gulliver2005} studied the (suitably generalised) merit factor of $d$-dimensional arrays of size $2\times 2\times\cdots\times 2$. In~\cite{Gulliver2005}, the merit factor of several families of such arrays was calculated. In particular, the largest asymptotic value, as $d\to\infty$, of the merit factor of a family of $d$-dimensional arrays considered in~\cite{Gulliver2005} equals~$3$.


\section*{Acknowledgements} 
I would like to thank Jonathan Jedwab for valuable discussions on the subject and for careful comments on a draft of this paper.


\end{document}